\documentclass[journal]{IEEEtran}
\usepackage{amsmath}
\usepackage{amssymb}
\usepackage{cite}
\usepackage{graphics}
\usepackage{graphicx}
\usepackage{subfigure}
\usepackage{epsfig}

\newtheorem{theorem}{Theorem}

\newtheorem{lemma}[theorem]{Lemma}

\DeclareMathOperator{\var}{var}
\DeclareMathOperator{\vect}{vec}
\DeclareMathOperator{\cov}{cov}
\DeclareMathOperator{\Ber}{Ber}

\begin{document}

\title{\textbf{Variance Analysis of Randomized Consensus in Switching
Directed~Networks}}
\author{Victor M. Preciado, Alireza Tahbaz-Salehi, and Ali Jadbabaie \thanks{%
This research is supported in parts by the following grants: DARPA/DSO
SToMP, NSF ECS-0347285, ONR MURI N000140810747, and AFOSR: Complex Networks
Program.} \thanks{%
Victor Preciado is with Department of Electrical and Systems Engineering,
University of Pennsylvania, Philadelphia, PA 19104 (e-mail:
preciado@seas.upenn.edu)} \thanks{%
Alireza Tahbaz-Salehi is with Department of Economics and Department of
Electrical and Systems Engineering, University of Pennsylvania,
Philadelphia, PA 19104 (e-mail: atahbaz@seas.upenn.edu).} \thanks{%
Ali Jadbabaie is with the General Robotics, Automation, Sensing and
Perception (GRASP) Laboratory, Department of Electrical and Systems
Engineering, University of Pennsylvania, Philadelphia, PA 19104 (e-mail:
jadbabai@seas.upenn.edu).}}
\maketitle

\begin{abstract}
In this paper, we study the asymptotic properties of distributed consensus
algorithms over switching directed random networks. More specifically, we
focus on consensus algorithms over independent and identically distributed,
directed Erd\H{o}s-R\'{e}nyi random graphs, where each agent can communicate
with any other agent with some exogenously specified probability $p$. While
it is well-known that consensus algorithms over Erd\H{o}s-R\'{e}nyi random
networks result in an asymptotic agreement over the network, an analytical
characterization of the distribution of the asymptotic consensus value is
still an open question. In this paper, we provide closed-form expressions
for the mean and variance of the asymptotic random consensus value, in terms
of the size of the network and the probability of communication $p$. We also
provide numerical simulations that illustrate our results.
\end{abstract}



\section{Introduction}

Due to their wide range of applications, distributed consensus algorithms
have attracted a significant amount of attention in the past few years.
Besides their applications in distributed and parallel computation \cite%
{TsitsiklisThesis}, distributed control \cite{Jad2003}, and robotics \cite%
{Bullo2005}, they have also been used as models of opinion dynamics and
belief formation in social networks \cite{DeGroot1974, Jackson_Naive}. The
central focus in this vast body of literature is to study whether a group of
agents in a network, with \textit{local} communication capabilities can
reach a \textit{global} agreement, using simple, deterministic information
exchange protocols.\footnote{%
For a survey on the most recent works in this area see \cite{Olfati2007}.}

More recently, there has also been some interest in understanding the
behavior of consensus algorithms in random settings \cite%
{Hatano2004,WuRandomConsensus,Porfiri_TAC,Tahbaz_Jad_TAC, Picci_Taylor_CDC,
Asu_misinformation}. The randomness can be either due to the choice of a
randomized network communication protocol or simply caused by the potential
unpredictability of the environment in which the distributed consensus
algorithm is implemented \cite{Zampieri}. It is recently shown that
consensus algorithms over i.i.d. random networks lead to a global agreement
on a possibly random value, as long as the network is connected in
expectation \cite{Tahbaz_Jad_TAC}.

While different aspects of consensus algorithms over random switching
networks, such as conditions for convergence \cite{Hatano2004,
WuRandomConsensus, Porfiri_TAC, Tahbaz_Jad_TAC} and the speed of convergence 
\cite{Zampieri}, have been widely studied, a characterization of the
distribution of the asymptotic consensus value has attracted little
attention. Two notable exceptions are Boyd \textit{et al.} \cite{Boyd2005b},
who study the asymptotic behavior of the random consensus value in the
special case of symmetric networks, and Tahbaz-Salehi and Jadbabaie \cite%
{Tahbaz_Jad_TAC_2}, who compute the mean and variance of the consensus value
for general i.i.d. graph processes. Nevertheless, a complete
characterization of the distribution of the asymptotic value for general 
\textit{asymmetric} random consensus algorithms remains an open problem.

In this paper, we study asymptotic properties of consensus algorithms over a
general class of switching, directed random graphs. More specifically, we derive closed-form
expressions for the mean and variance of the asymptotic consensus value,
when the underlying network evolves according to an i.i.d. \textit{directed}
Erd\H{o}s-R\'{e}nyi random graph process. In our model, at each time period,
a directed communication link is established between two agents with some
exogenously specified probability $p$. It is well-known that due to the
connectivity of the expected graph, consensus algorithms over Erd\H{o}s-R%
\'{e}nyi random graphs result in asymptotic agreement. However, due to the
potential asymmetry in pairwise communications between different agents, the
asymptotic value of consensus is not guaranteed to be the average of the
initial conditions. Instead, agents will asymptotically agree on some random
value in the convex hull of the initial conditions. Our closed-form
characterization of the variance provides a quantitative measure of how
dispersed the random agreement point is around the average of the initial
conditions in terms of the fundamentals of the model, namely, the size of
the network and the exogenous probability of communication $p$.

The rest of the paper is organized as follows. In the next section, we
describe our model of random consensus algorithms. 
In Section \ref{Section_Quad_Analysis}, we derive an explicit expression for
the variance of the limiting consensus value over switching directed Erd\H{o}%
s-R\'{e}nyi random graphs in terms of the size of the network $n$ and the
communication probability $p$. Section IV contains simulations of our
results and Section V concludes the paper.


\section{Consensus Over Switching Random Graphs}

\label{Section_Model}

Consider the discrete-time linear dynamical system 
\begin{equation}  \label{consensus_update}
\mathbf{x}\left( k\right) =W_{k}\mathbf{x}\left( k-1\right) ,
\end{equation}
where $k\in \left\{ 1,2,\dots\right\} $ is the discrete time index, $\mathbf{%
x}(k) \in \mathbb{R}^{n}$ is the state vector at time $k$, and $%
\{W_{k}\}_{k=1}^\infty$ is a sequence of stochastic matrices. We interpret (%
\ref{consensus_update}) as a distributed scheme where a collection of
agents, labeled 1 through $n$, update their state values as a convex
combination of the state values of their neighbors at the previous time
step. Given this interpretation, $\mathbf{x}_i(k)$ corresponds to the state
value of agent $i$ at time $k$, and $W_k$ captures the neighborhood relation
between different agents at time $k$: the $ij$ element of $W_k$ is positive
only if agent $i$ has access to the state of agent $j$. For the remainder of
the paper, we assume that the weight matrices $W_k$ are randomly generated
by an independent and identically distributed matrix process.

We say dynamical system (\ref{consensus_update}) reaches \textit{consensus}
asymptotically on some path $\{W_k\}_{k=1}^\infty$, if along that path,
there exists $x^*\in\mathbb{R}$ such that $\mathbf{x}_i(k)\rightarrow x^*$
for all $i$ as $k\rightarrow\infty$. We refer to $x^*$ as the \textit{%
consensus value}. It is well-known that for i.i.d. random networks dynamical
system (\ref{consensus_update}) reaches consensus on almost all paths if and
only if the graph corresponding to the communications between agents is
connected in expectation. More precisely, Tahbaz-Salehi and Jadbabaie \cite%
{Tahbaz_Jad_TAC} show that $W_k\dots W_2W_1 \longrightarrow \mathbf{1}d^T$
almost surely $-$ where $d$ is some random vector $-$ if and only if the
second largest eigenvalue modulus of $\mathbb{E}W_k$ is subunit. Clearly,
under such conditions, dynamical system (\ref{consensus_update}) reaches
consensus with probability one where the consensus value is a random
variable equal to $x^*=d^T\mathbf{x}(0)$, where $\mathbf{x}(0)$ is the
vector of initial conditions.

A complete characterization of the random consensus value $x^{\ast }$ is an
open problem. However, it is possible to compute its mean and variance in
terms of the first two moments of the i.i.d. weight matrix process. In \cite%
{Tahbaz_Jad_TAC_2}, the authors prove that the conditional mean of the
random consensus value is given by the random consensus value are given by 
\begin{equation*}
\mathbb{E}x^{\ast }=\mathbf{x}(0)^{T}\mathbf{v}_{1}(\mathbb{E}W_{k}),
\end{equation*}%
and its conditional variance is equal to 
\begin{eqnarray}
\lefteqn{\var(x^{\ast })=\left[ \mathbf{x}(0)\otimes \mathbf{x}(0)\right]
^{T}\vect(\cov(d))}  \label{variance_expression} \\
&&\!\!\!\!\!=\left[ \mathbf{x}(0)\otimes \mathbf{x}(0)\right] ^{T}\mathbf{v}%
_{1}(\mathbb{E}\left[ W_{k}\otimes W_{k}\right] )-[\mathbf{x}(0)^{T}\mathbf{v%
}_{1}(\mathbb{E}W_{k})]^{2}  \notag
\end{eqnarray}%
where $\mathbf{v}_{1}\left( \cdot \right) $ denotes the normalized left
eigenvector corresponding to the unit eigenvalue, and $\otimes $ denotes the
Kronecker product. In the following, we shall use (\ref{variance_expression}%
) to derive an explicit expression for the mean and variance of the
consensus value over a class of switching, directed random graphs..




\section{Variance Analysis for Finite Erd\H{o}s-R\'{e}nyi Random Graphs}

\label{Section_Quad_Analysis}

\subsection{Directed Erd\H{o}s-R\'{e}nyi Random Graphs}

We consider directed graphs $G=\left( V,E\right) $ with a fixed set of
vertices $V=\left\{ 1,...,n\right\} $ and directed edges. A directed edge
from vertex $i$ to vertex $j$ is representes as an ordered pair $\left(
i,j\right) $, with $i,j\in V$. In a \textit{directed} Erd\H{o}s-R\'{e}nyi
(ER) graph $\mathcal{G}\left( n,p\right) $, the existence of a directed edge 
$\left( i,j\right) ,$ with $i\neq j,$ is determined randomly and
independently of other edges with a fixed probability $p\in \left[ 0,1\right]
$. The adjacency matrix $A=\left[ a_{ij}\right] $ associated to $\mathcal{G}%
\left( n,p\right) $ is a random matrix with all zeros in the diagonal, and
off-diagonal elements $a_{ij}=1$ with probability $p$, and $0$ with
probability $1-p$, for $i\neq j$. Similarly, the out-degree matrix can be
defined from the adjacency matrix as $D=$diag$\left\{ d_{i}\right\} $, where 
$d_{i}=\sum_{j}a_{ij}$.

In what follows, we associate a sequence of stochastic matrices $%
\{W_{k}\}_{k=1}^{\infty }$ to a sequence of i.i.d random realizations of
directed ER graphs $\left\{ \mathcal{G}_{k}\left( n,p\right) \right\}
_{k=1}^{\infty }$. To each random graph realization, we associate the
following stochastic matrix:%
\begin{equation}
W_{k}=\left( D_{k}+I_{n}\right) ^{-1}\left( A_{k}+I_{n}\right) ,
\label{Stochastic Matrix}
\end{equation}%
where $A_{k}$ and $D_{k}$ are the adjacency and out-degree matrices of the
graph realization. Notice that adding the identity matrix to the adjacency
in (\ref{Stochastic Matrix}) is equivalent to introduce a self-loops (an
edge that starts and ends at the same vertex) over every single vertex in $V$%
. These self-loops serve to avoid singularities associated with the presence
of isolated nodes in $\mathcal{G}_{k}\left( n,p\right) $ (for which $d_{i}=0$%
, and $D_{k}$ is not invertible).


\subsection{\label{Start Computations}Variance of Consensus Value}

In this section, we derive an explicit expression for the variance of the
limiting consensus value for a switching directed random graph. We base our
analysis in studying the terms in (\ref{variance_expression}), i.e., $%
\mathbf{v}_{1}\left( \mathbb{E}W_{k}\right) $ and $\mathbf{v}_{1}\left( 
\mathbb{E}\left[ W_{k}\otimes W_{k}\right] \right) $. In order to compute $%
\mathbf{v}_{1}\left( \mathbb{E}W_{k}\right) $, we first compute the
expectation of each entry in $W_{k}$. The expectation of the diagonal
entries of $W_{k}$ are equal to $\mathbb{E}\left[ 1/\left( d_{i}+1\right) %
\right] $, where $d_{i}$ is a random variable representing the degree of the 
$i$-th node. In a random ER graph with probability of link $p$ (and
complement $q\triangleq 1-p)$, the probability density of $d_{i}$ is a
Bernoulli distribution with $n-1$ trials and parameter $p$, i.e., $f\left(
d_{i}\right) \sim \Ber\left( n-1,p\right) $. Hence, 
\begin{eqnarray}
\mathbb{E}w_{ii} &=&\mathbb{E}\left[ \frac{1}{d_{i}+1}\right]
=\sum_{k=0}^{n-1}\frac{1}{k+1}\binom{n-1}{k}p^{k}q^{n-k-1}  \notag \\
&=&\frac{1-q^{n}}{np}\triangleq f_{1}\left( p,n\right) ,  \label{Diagonal}
\end{eqnarray}%
where we have defined $f_{1}$ for future convenience. Furthermore, the
off-diagonal elements of $\mathbb{E}W_{k}$ are equal to 
\begin{equation*}
\mathbb{E}w_{ij}=\mathbb{E}\left[ \frac{a_{ij}}{d_{i}+1}\right] =\mathbb{E}%
\left[ \left. \frac{1}{d_{i}+1}\right\vert a_{ij}=1\right] \mathbb{P}\left(
a_{ij}=1\right) ,
\end{equation*}%
where we have applied the law of total expectation in the last equality.
Moreover, it is straightforward to show that $f(d_{i}-1|a_{ij}=1)\sim \Ber%
(n-2,p)$; thus, 
\begin{eqnarray}
\mathbb{E}w_{ij} &=&p\sum_{k=0}^{n-2}\frac{1}{k+2}\binom{n-2}{k}%
p^{k}q^{n-k-2}  \notag \\
&=&\frac{q^{n}+np-1}{np\left( n-1\right) }=\frac{1-f_{1}\left( p,n\right) }{%
n-1}.  \label{Off-diagonals}
\end{eqnarray}%
Taking (\ref{Diagonal}) and (\ref{Off-diagonals}) into account, we can write 
$\mathbb{E}W_{k}$ as follows: 
\begin{eqnarray*}
\mathbb{E}W_{k} &=&\mathbb{E}w_{ij}\mathbf{1}_{n}\mathbf{1}_{n}^{T}+\left( 
\mathbb{E}w_{ii}-\mathbb{E}w_{ij}\right) I_{n} \\
&=&\frac{1-f_{1}\left( p,n\right) }{n-1}\mathbf{1}_{n}\mathbf{1}_{n}^{T}-%
\frac{1-n~f_{1}\left( p,n\right) }{n-1}I_{n}.
\end{eqnarray*}%
It is easy to verify that $\mathbb{E}W_{k}$ is irreducible. Therefore, as
discussed in the previous section, consensus algorithms over i.i.d. directed
Erd\H{o}s-R\'{e}nyi random graph process converge to consensus with
probability one. Moreover, it is straightforward to show that vector $%
\mathbf{1}_{n}$ satisfies the eigenvalue equation $\mathbf{1}_{n}^{T}~%
\mathbb{E}W_{k}=$ $\mathbf{1}_{n}^{T}$; thus, 
\begin{equation}
\mathbf{v}_{1}\left( \mathbb{E}W_{k}\right) =\frac{1}{n}\mathbf{1}_{n}.
\label{Small Eigenvector}
\end{equation}%
Therefore, as expected, the mean of the random consensus value is equal to
the average of $\mathbf{x}\left( 0\right) $, i.e., 
\begin{equation}
\mathbb{E}x^{\ast }=\frac{1}{n}\sum_{i=1}^{n}x_{i}\left( 0\right) \triangleq 
\bar{x}\left( 0\right) .  \label{Asymptotic Expectation}
\end{equation}

The other term in the expression of variance that we need to compute is $%
\mathbf{v}_{1}(\mathbb{E}\left[ W_{k}\otimes W_{k}\right])$. In order to
compute this vector, we first compute the entries of matrix $\mathbb{E}\left[
W_{k}\otimes W_{k}\right]$, which are of the form $\mathbb{E}(w_{ij}w_{rs})$%
, with $i, j, r$ and $s$ ranging from $1$ to $n$. The entries can be
classified into six different cases depending on the relations between the
indices. Below, we present the expressions for each case. Some of the
expressions are in terms of the \textit{hypergeometric} function $_{3}
F_2(1,1,1-n;2,2;p/(p-1))$, which for convenience we denote by $H(p,n)$,
defined as the power series 
\begin{equation*}
H(p,n) = \sum_{k=0}^{n-1} \frac{1}{(k+1)^2}\binom{n-1}{k} \left(\frac{p}{1-p}%
\right)^{k}.
\end{equation*}
In the following expressions we assume that all four indices $i,j,r$ and $s$
are distinct. Detailed computations are provided in the Appendix.

\begin{eqnarray}
Q_1&=&\mathbb{E}(w_{ii}^2) = q^{n-1} H(p,n),  \label{expectation cases} \\
Q_2&=&\mathbb{E}(w_{ii}w_{jj}) = f_{1}^{2}(p,n),  \notag \\
Q_3&=&\mathbb{E}(w_{ii}w_{is}) = \mathbb{E}(w_{ij}w_{ii}) = \mathbb{E}%
(w_{ij}^2)  \notag \\
& = & \frac{f_1(p,n)-q^{n-1}H(p,n)}{n-1},  \notag \\
Q_4&=&\mathbb{E}(w_{ii}w_{ri}) = \mathbb{E}(w_{ii}w_{rs})=f_1(p,n) \frac{%
1-f_1(p,n)}{n-1}  \notag \\
Q_5&=&\mathbb{E}(w_{ij}w_{is}) = \frac{1+2q^{n-1}H(p,n)-3f_1(p,n)}{(n-1)(n-2)%
},  \notag \\
Q_6&=&\mathbb{E}(w_{ij}w_{ji}) = \mathbb{E}(w_{ij}w_{js})= \mathbb{E}%
(w_{ij}w_{ri})  \notag \\
& = & \mathbb{E}(w_{ij}w_{rj}) = \mathbb{E}(w_{ij}w_{rs}) = \left( \frac{%
1-f_1(p,n)}{n-1}\right)^2.  \notag
\end{eqnarray}

As stated earlier, every entry of $\mathbb{E}[W_k\otimes W_k]$ is equal to
one of the expressions provided in (\ref{expectation cases}). The key
observation is the pattern that such classification of entries induces in
the matrix. In order to clarify this point, we illustrate this pattern for $%
n=3$, where the numbers in parenthesis correspond to one of the six cases
identified above. As the figure suggests, the entries of $\mathbb{E}[
W_{k}\otimes W_{k}] $ are identical to the entries of $K\otimes K$ except
for those at the rows $1+r(n+1)$ for $r=0,...,n-1$, where matrix $K\in%
\mathbb{R}^{n\times n}$ is defined as 
\begin{equation*}
K = \frac{1-f_1(p,n)}{n-1}\mathbf{1}_{n}\mathbf{1}_{n}^T+\frac{nf_{1}(p,n)-1%
}{n-1}I_{n},
\end{equation*}
with $f_{1}(p,n)$ as its diagonal, and $[1-f_1(p,n)]/(n-1)$ as its
off-diagonal entries. In Fig. 1, the entries of $\mathbb{E}[W_k\otimes W_k]$
that are different from the entries of $K\otimes K$ are marked with bold
lines.

\begin{figure}
 \centering
 \includegraphics[width=0.75\linewidth]{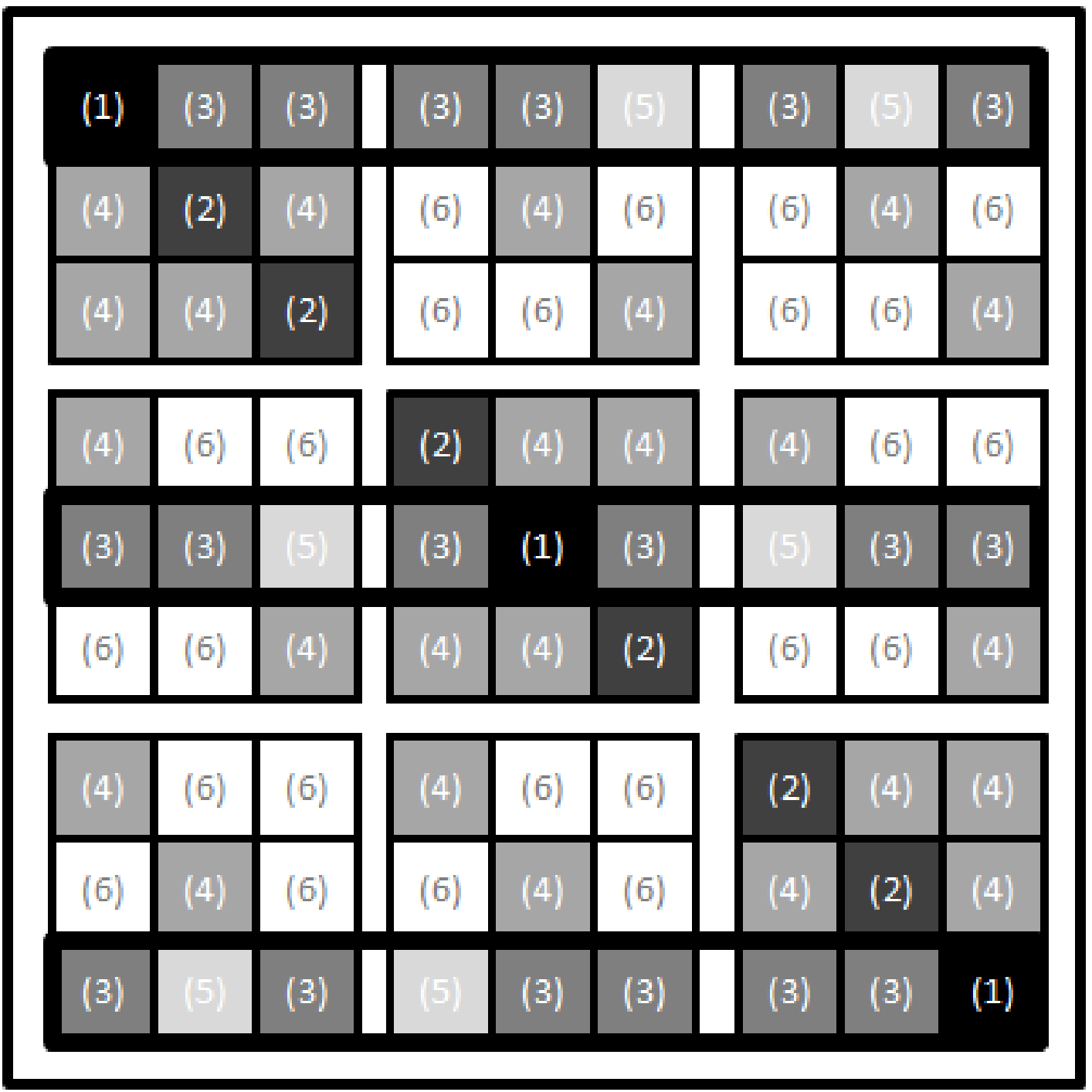}
 \caption{The pattern of $\mathbb{E}[W_k\otimes W_k]$ for $n=3$. The numbers in parentheses
represent the value of each entry in terms of expressions $Q1$ to $Q6$ defined
in (8). The entries that are different from the corresponding entries of $K\otimes K$ are marked with bold lines.}
\end{figure}

We now exploit the identified pattern to explicitly compute the left
eigenvector $\mathbf{v}_{1}(\mathbb{E}[W_{k}\otimes W_{k}])$.

\begin{lemma}
The left eigenvector of $\mathbb{E}\left[ W_{k}\otimes W_{k}\right]$
corresponding to its unit eigenvalue is given by 
\begin{equation}  \label{Big Eigenvector}
\mathbf{v}_{1}(\mathbb{E}[W_{k}\otimes W_{k}]) =\frac{1}{\delta}\left[\rho (%
\mathbf{1}_{n}\mathbf{\otimes 1}_{n}) +(1-\rho) \sum_{i=1}^{n}(\mathbf{e}%
_{i}\otimes \mathbf{e}_{i})\right]
\end{equation}
where $\rho$ and $\delta$ depend on $p$ and $n$ as follows: 
\begin{eqnarray}
\rho(p,n) &\triangleq & \frac{p(n-1)}{p(n-2)+1-(1-p)^n},
\label{Explicit Ratio} \\
\delta(p,n) &\triangleq & n+n(n-1)\rho(p,n).  \label{Explicit Sigma}
\end{eqnarray}
\end{lemma}

\vspace{0.1in} 
\begin{proof}
First of all, notice that $\mathbb{E}[W_k\otimes W_k]$ is a stochastic matrix whose entries are all strictly positive for $p>0$. Therefore, it has a unique left eigenvector corresponding to its unit eigenvalue, which means that $\mathbf{v}_1(\mathbb{E}[W_k\otimes W_k])$ is well-defined. We now show that the pattern of this left eigenvector is of the form
\begin{equation}\label{Proposed Eigenvector}
\mathbf{v} = \alpha (\mathbf{1}_{n}\mathbf{\otimes 1}_{n}) +(\beta-\alpha ) \sum_{i=1}^{n}(\mathbf{e}_{i}\otimes \mathbf{e}_{i}) ,  
\end{equation}
for some positive numbers $\alpha$ and $\beta$. Notice that all entries of this vector are equal to $\alpha$, except for the ones indexed $1+r( n+1)$ for $r=0,...,n-1$, which are equal to $\beta$. To show that the eigenvector we are looking for is indeed of the given pattern, we premultiply $\mathbb{E}[W_{k}\otimes W_{k}] $ by $\mathbf{v}$, and verify that
\begin{equation}\label{Big Linear Map}
\mathbf{v}^{T}\mathbb{E}[ W_{k}\otimes W_{k}] =\alpha'(\mathbf{1\otimes 1})^{T}+(\beta'-\alpha') \sum_{i=1}^{n}(\mathbf{e}_{i}\otimes\mathbf{e}_{i})^{T}.  
\end{equation}
where $\alpha'$ and $\beta'$ are positive numbers, given by 
\begin{equation}\label{albet}
\begin{bmatrix}
\alpha'\\
\beta'
\end{bmatrix}=\begin{bmatrix}
A & B\\
C & D
\end{bmatrix}\begin{bmatrix}
\alpha\\ \beta
\end{bmatrix}
\end{equation}
with coefficients $A$, $B$, $C$, and $D$ defined as
\begin{eqnarray*}
A & = & 1+ [nf_1(p,n)+n-2]\frac{1-f_1(p,n)}{(n-1)^2}  \\
B & = & \frac{1-f_1(p,n)}{n-1}\\
C & = & [nf_1(p,n)+n-2]\frac{1-f_1(p,n)}{n-1}\\
D & = & f_1(p,n)
\end{eqnarray*}
Equation (\ref{Big Linear Map}) suggests that the pattern of $\mathbf{v}_1$ is preserved when it is multiplied by $\mathbb{E}[W_k\otimes W_k]$. Therefore, the vector defined in (\ref{Proposed Eigenvector}) is the unique left eigenvector of the matrix if there are positive numbers $\alpha=\alpha'$ and $\beta=\beta'$ that satisfy (\ref{albet}). 

Due to the fact $BC=(1-A)(1-D)$, the matrix in (\ref{albet}) has an eigenvalue equal to one, with eigenvector $\begin{bmatrix}B & 1-A\end{bmatrix}^T$, implying that such $(\alpha,\beta,\alpha',\beta')$ exist. Thus, the proposed vector in (\ref{Proposed Eigenvector}) is an eigenvector of $\mathbb{E}[W_k\otimes W_k]$ as long as $\alpha=\beta B/(1-A)=\rho(p,n)\beta$, which means that
\begin{equation*}
\mathbf{v}_{1}(\mathbb{E}[W_{k}\otimes W_{k}]) =\frac{1}{\delta}\left[\rho (\mathbf{1}_{n}\mathbf{\otimes 1}_{n}) +(1-\rho) \sum_{i=1}^{n}(\mathbf{e}_{i}\otimes \mathbf{e}_{i})\right]. 
\end{equation*}
Note that $\delta(p,n)$, defined in (\ref{Explicit Sigma}), is a normalizing factor guaranteeing that the elements of the vector sum up to one.
\end{proof}

Now that we have derived explicit expressions for the eigenvectors (\ref%
{Small Eigenvector}) and (\ref{Big Eigenvector}), we can compute a
closed-form expression for the variance of the limiting consensus value in
terms of $p$ and $n$.

\begin{theorem}
\label{Variance Expression} The variance of the asymptotic consensus value $%
x^{\ast }$ of the distributed update defined in (\ref{consensus_update})
over switching Erd\H{o}s-R\'{e}nyi random graphs with parameter $p$ is given
by 
\begin{equation}
\var(x^{\ast })=\frac{1-\rho }{\delta }\sum_{i=1}^{n}\left[ x_{i}(0)-\bar{x}%
(0)\right] ^{2},  \label{Explicit Variance}
\end{equation}%
where $\rho (p,n)$ and $\delta (p,n)$ are defined in (\ref{Explicit Ratio})
and (\ref{Explicit Sigma}), respectively.
\end{theorem}

\begin{proof}
First, from (\ref{Small Eigenvector}), we have that
\begin{equation*}
\left[\mathbf{x}(0) ^{T}\mathbf{v}_{1}(\mathbb{E}W_{k})\right] ^{2}=\left(\frac{1}{n}\sum_{i=1}^{n}x_{i}(0) \right) ^{2}=\bar{x}(0) ^{2}.
\end{equation*}
On the other hand, from (\ref{Big Eigenvector}), we have that
\begin{eqnarray*}
\lefteqn{\left[\mathbf{x}(0) \otimes \mathbf{x}(0)\right]^{T} \mathbf{v}_{1}\left( \mathbb{E}[ W_{k}\otimes W_{k}] \right) =} &&\\
&=&\frac{1}{\delta}\left[\mathbf{x}(0) \otimes \mathbf{x}(0)\right]^{T} \left[\rho (\mathbf{1}_{n}\mathbf{\otimes 1}_{n}) +(1-\rho) \sum_{i=1}^{n}(\mathbf{e}_{i}\otimes \mathbf{e}_{i})\right]\\
&=&\frac{\rho }{\delta }[\mathbf{x}(0)^{T}\mathbf{1}_{n}]\otimes [\mathbf{x}(0)^{T}\mathbf{1}_{n}]  \\
&&+\frac{1-\rho}{\delta}\sum_{i=1}^{n}\left([\mathbf{x}(0)^{T}\mathbf{e}_{i}]\otimes [\mathbf{x}(0)^{T}\mathbf{e}_{i}]\right) ,
\end{eqnarray*}
where we have used the fact that $(A\otimes B) (C\otimes D) =AC\otimes BD $. Since the Kronecker terms in the last expression are scalars, we have
\begin{equation*}
[\mathbf{x}(0) \otimes \mathbf{x}(0)]^{T} \mathbf{v}_{1}(\mathbb{E}[ W_{k}\otimes W_{k}]) =\frac{\rho}{\delta}n^{2}\bar{x}(0)^{2} +\frac{1-\rho}{\delta} \sum_{i=1}^{n} x^2_i(0)
\end{equation*}
and therefore,
\begin{equation*}
\var(x^*) =\left(\frac{\rho}{\delta}n^{2}-1\right) [\bar{x}(0)]^{2}+\frac{1-\rho }{\delta}\sum_{i=1}^{n} [x_{i}(0)] ^{2}.
\end{equation*}
By adding and subtracting $n\frac{1-\rho}{\delta }[\bar{x}\left( 0\right)]^{2}$, the expression for the variance can be rewritten as
\begin{eqnarray*}
\var x^* &=& \frac{n(n-1)\rho+n-\delta}{\delta} [\bar{x}(0)]^{2} \\
&&+\frac{1-\rho}{\delta}\sum_{i=1}^{n}\left( x_{i}(0) -\bar{x}(0) \right) ^{2}.
\end{eqnarray*}
Since $\delta = n+n(n-1)\rho$, as defined in (\ref{Explicit Sigma}), the first term in the right-hand-side of the above expression is equal to zero. This proves the theorem.
\end{proof}

Expression (\ref{Explicit Variance}) shows that, given the parameters of the
random graph process $p$ and $n$, the variance of the limiting consensus
value, $x^{\ast }$, is equal to the empirical variance of the initial
conditions multiplied by the factor $n(1-\rho )/\delta $, which only depends
on parameters $p$ and $n$.

\bigskip 

\section{Numerical Simulations}

In this subsection, we present several simulations that illustrate the
result in Theorem \ref{Variance Expression}. In our first simulation, we
compare the analytical expression for the variance in (\ref{Explicit
Variance}) with the empirical variance obtained from 100 realizations of the
random consensus algorithm for $n$ in a certain range.
In our simulations, we compute the (analytical and empirical) variances for
a range of network sizes while keeping the expected out-degree of the random
graphs fixed to a constant value $c$ (i.e., the probability communication in the
random graphs is then $p=c/n$, for all $n$). In Fig. 1, we plot both the
analytical and empirical variances when the network sizes $n$ goes from 5 to
50 nodes and the expected degree is fixed to be $c=5,$ for all $n$. The
initial conditions for each network size is given by $x_{i}\left( 0\right)
=i/n$, for $i=1,...,n$.

\bigskip

\begin{figure}
 \centering
 \includegraphics[width=0.95\linewidth]{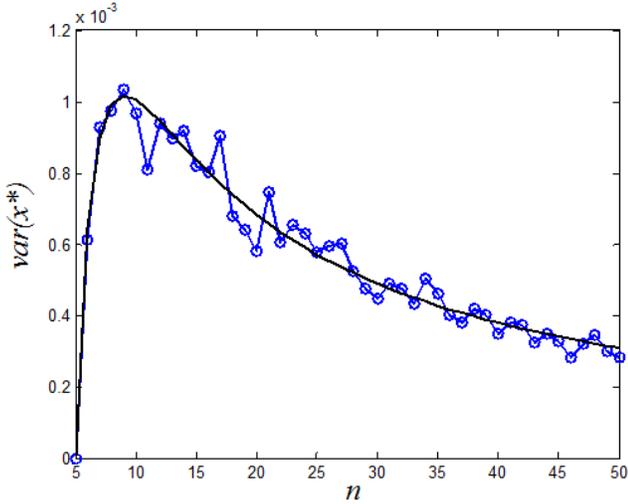}
 \caption{Comparison between the empirical variance and the analytical variance for $n$ in the interval [5:50] and $p=5/n$ for all $n$.}
\end{figure}

\bigskip

Several comments are in order about the behavior of $\var\left( x^{\ast
}\right) $ in Fig. 1. First, for $n=c=5$, we have that $p=c/n=1$ (every link
exist); hence, the random graph is not random, but a complete graph $K_{n}$,
and the distributed consensus algorithm converges to the average of the
initial conditions with zero variance (as one can check in Fig. 1). Second,
when $n$ is slightly over $c$, the variance increases quite abruptly with $n$
until it reaches a maximum value. For $c=5$, this maximum is achieved for a
network size of $9$ nodes. The location of this interesting point, that we
denote by $\hat{n}\left( c\right) $, can be easily computed using Theorem %
\ref{Variance Expression}. Furthermore, for $n>$ $\hat{n}(c)$ the variance
slowly decreases with the network size. One can prove that this variance
tends asymptotically to zero as $n\rightarrow \infty $ at a rate $1/n$.

Furthermore, according to (\ref{Explicit Variance}), given a vector of $%
\mathbf{x}\left( 0\right) $ initial condition, the variance of the
asymptotic consensus value is equal to the empirical variance of the entries
of $\mathbf{x}\left( 0\right) $ rescaled by the factor $n\left( 1-\rho
\right) /\delta $ (where $\rho $ and $\delta $ depend on the random graph
parameters $p$ and $n$). In Fig. 2, we plot the values of the factor $%
n\left( 1-\rho \right) /\delta $ for a set of expected degrees $c\in \left\{
5,6,7,8,9,10\right\} $ while the network size $n$ varies from $5$ to $70$.
We observe that the behavior of the factor $n\left( 1-\rho \right) /\delta $
is similar for any given $c$. Again, for $n=c$ the variances are zero and the
variances grow abruptly until a maximum, $\hat{n}\left( c\right) $, is
reached. For large values of $n$, the variance slowly decays towards zero at
a rate $1/n$ for all $c$.

\bigskip

\begin{figure}
 \centering
 \includegraphics[width=0.95\linewidth]{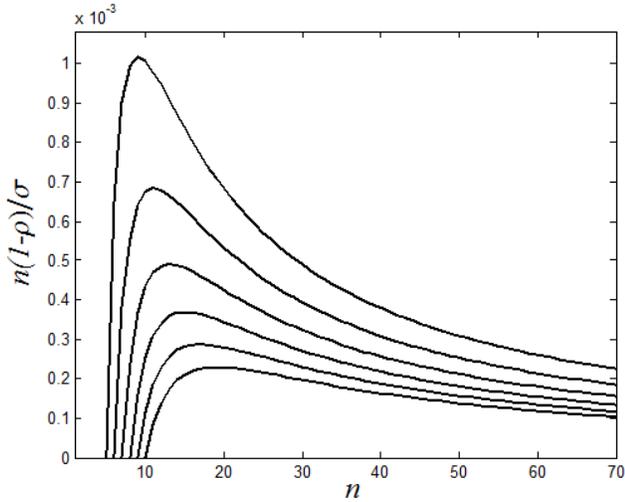}
 \caption{Several analytical variance for $c\in \left\{ 5,6,7,8,9,10\right\} $.}
\end{figure}

\bigskip

\section{Conclusions and Future Work}

We have studied the asymptotic properties of the consensus value
in distributed consensus algorithms over switching,
directed random graphs. Due to the connectivity of the expected graph,
consensus algorithms over Erd\H{o}s-R\'{e}nyi random graphs result in
asymptotic agreement. However, the asymptotic value of consensus is not
guaranteed to be the average of the initial conditions. Instead, agents will
asymptotically agree on some random value in the convex hull of the initial
conditions. While different aspects of consensus algorithms over random
switching networks, such as conditions for convergence and the speed of
convergence, have been widely studied, a characterization of the
distribution of the asymptotic consensus for general \textit{asymmetric}
random consensus algorithms remains an open problem.

In this paper, we have derived closed-form expressions for the expectation
and variance of the asymptotic consensus value as functions of the number of
nodes $n$ and the probability of existence of a communication link, $p$.
While the expectation of the distribution of the consensus value is simply
the mean of the initial conditions over the nodes of the network, the
variance presents an interesting structure. In particular, the variance of
this limiting distribution of the consensus value is equal to the empirical variance of the set of initial
conditions multiplied by a factor that depends on $p$ and $n$. We have derived an explicit expression for this factor and check its
validity with numerical simulations.

\bigskip 

\appendix[Kronecker Matrix Entries]

The Appendix contains the detailed computation of the entries of matrices $%
\mathbb{E}W_{k}$ and $\mathbb{E}[W_{k}\otimes W_{k}]$. We start by computing
the elements of $\mathbb{E}W_{k}$. The diagonal entries of $\mathbb{E}W_{k}$
are given by: 
\begin{eqnarray*}
\mathbb{E}w_{ii} &=&\mathbb{E}\left[ \frac{1}{1+d_{i}}\right]
=\sum_{k=0}^{n-1}\frac{1}{k+1}\mathbb{P}(d_{i}=k) \\
&=&\sum_{k=0}^{n-1}\frac{1}{k+1}\binom{n-1}{k}p^{k}(1-p)^{n-k-1} \\
&=&\frac{1-q^{n}}{np}\triangleq f_{1}(p,n)
\end{eqnarray*}%
On the other hand, the non-diagonal entries of $\mathbb{E}W_{k}$ result in:%
\begin{equation*}
\mathbb{E}w_{ij}=\frac{1}{n-1}[1-\mathbb{E}w_{ii}]=\frac{np-1+q^{n}}{np(n-1)}%
=\frac{1-f_{1}(p,n)}{n-1}
\end{equation*}

We now turn to the computation of the elements of $\mathbb{E}[W_{k}\otimes
W_{k}]$, which are of the form $\mathbb{E}(w_{ij}w_{rs})$. In what follows
we assume that the indices $i$, $j$, $r$, and $s$ are distinct. We first
start with elements with in the diagonal subblocks of $\mathbb{E}%
[W_{k}\otimes W_{k}]$:

\begin{eqnarray*}
\mathbb{E}w_{ii}^{2} &=&\mathbb{E}\left[ \frac{1}{(d_{i}+1)^{2}}\right]  \\
&=&\sum_{k=0}^{n-1}\frac{1}{(k+1)^{2}}\binom{n-1}{k}p^{k}q^{n-k-1} \\
&=&q^{n-1}H(p,n)
\end{eqnarray*}%
For the rest of entries in the diagonal blocks, it is useful to note that $%
\mathbb{E}w_{ii}=f_{1}(p,n)$ and $\mathbb{E}w_{ij}=\frac{1-f_{1}(p,n)}{n-1}$%
, as proved in Subsection \ref{Start Computations}:%
\begin{eqnarray*}
\mathbb{E}(w_{ii}w_{rr}) &=&\mathbb{E}w_{ii}~\mathbb{E}w_{rr}=f_{1}^{2}(p,n)
\\
&& \\
\mathbb{E}(w_{ii}w_{ri}) &=&\mathbb{E}w_{ii}~\mathbb{E}w_{ri}=f_{1}(p,n)%
\frac{1-f_{1}(p,n)}{n-1} \\
&& \\
\mathbb{E}(w_{ii}w_{rs}) &=&\mathbb{E}w_{ii}~\mathbb{E}w_{rs}=f_{1}(p,n)%
\frac{1-f_{1}(p,n)}{n-1} \\
&& \\
\mathbb{E}(w_{ii}w_{is}) &=&\frac{\mathbb{E}w_{ii}-\mathbb{E}w_{ii}^{2}}{n-1}%
=\frac{f_{1}(p,n)-q^{n-1}H(p,n)}{n-1} \\
&&
\end{eqnarray*}%
Similarly, we have the following results for the off-diagonal subblocks:

\begin{eqnarray*}
\mathbb{E}(w_{ij}w_{ji}) &=& \mathbb{E}w_{ij}~\mathbb{E}w_{ji} = \left(\frac{%
1-f_1(p,n)}{n-1}\right)^2 \\
&& \\
\mathbb{E}(w_{ij}w_{js}) &=& \mathbb{E}w_{ij}~\mathbb{E}w_{js} = \left(\frac{%
1-f_1(p,n)}{n-1}\right)^2 \\
&& \\
\mathbb{E}(w_{ij}w_{ri}) &=& \mathbb{E}w_{ij}~\mathbb{E}w_{ri} = \left(\frac{%
1-f_1(p,n)}{n-1}\right)^2 \\
&& \\
\mathbb{E}(w_{ij}w_{rj}) &=& \mathbb{E}w_{ij}~\mathbb{E}w_{rj} = \left(\frac{%
1-f_1(p,n)}{n-1}\right)^2 \\
&& \\
\mathbb{E}(w_{ij}w_{rs}) &=& \mathbb{E}w_{ij}~\mathbb{E}w_{rs} = \left(\frac{%
1-f_1(p,n)}{n-1}\right)^2 \\
&& \\
\mathbb{E}(w_{ij}w_{ii}) &=& \mathbb{E}(w_{ii}w_{is}) = \frac{%
f_1(p,n)-q^{n-1}H(p,n)}{n-1} \\
&& \\
\mathbb{E}(w_{ij}^2) & = & \mathbb{E}\left[\frac{a_{ij}^2}{(d_{i}+1)^2}%
\right] = \mathbb{E}\left[\frac{a_{ij}}{(d_{i}+1)^2}\right] \\
&=& \mathbb{E}(w_{ii}w_{ij}) = \frac{f_1(p,n)-q^{n-1}H(p,n)}{n-1} \\
&& \\
\mathbb{E}(w_{ij}w_{is}) & = & \frac{\mathbb{E}w_{ij}-\mathbb{E}w_{ij}^2-%
\mathbb{E}(w_{ij}w_{ii})}{n-2} \\
& = & \frac{1+2q^{n-1}H(p,n)-3f_1(p,n)}{(n-1)(n-2)} \\
&&
\end{eqnarray*}



\end{document}